\documentclass[conference,a4paper]{IEEEtran}
\usepackage{multicol}
\usepackage{array}

\usepackage{epsfig}
\usepackage{framed}
\usepackage[utf8]{inputenc}
\usepackage{amssymb}
\usepackage{amsthm}
\usepackage{amsmath}
\usepackage{amsfonts}
\usepackage{graphicx}
\usepackage{comment}
\usepackage{float}
\usepackage{subfigure}
\usepackage[ruled]{algorithm2e}

\allowdisplaybreaks

\usepackage{enumitem}
\usepackage{mathrsfs}
\usepackage{graphicx}
\usepackage{multirow}
\usepackage{color,soul}
\usepackage{color}
\usepackage{float}
\usepackage{enumitem}
\usepackage{adjustbox}
\usepackage{algorithmicx}
\usepackage{hyperref}
\usepackage{algcompatible}
\usepackage[noend]{algpseudocode}
\usepackage[outercaption]{sidecap}
\DeclareMathAlphabet{\mathpzc}{OT1}{pzc}{m}{it}
\algrenewcommand\alglinenumber[1]{\scriptsize #1:}
\definecolor{mygray}{gray}{0.6}

\newtheorem{definition}{Definition}
\newtheorem{lemma}{Lemma}
\newtheorem{theorem}{Theorem}

\newcommand{\pr}[1]{\mathrm{Pr}#1} 

\newcommand{\mc}[1]{\mathcal{#1}}
\newcommand{\ms}[1]{\mathsf{#1}}

\newcommand{\remove}[1]{}

\newcommand{\ikemg}{\mathsf{OWSKA}_{a}\mathsf{.Gen}}
\newcommand{\ikeme}{\mathsf{OWSKA}_{a}\mathsf{.Alice}}
\newcommand{\ikemd}{\mathsf{OWSKA}_{a}\mathsf{.Bob}}

\newcommand{\skp}{SKA$_p$~}
\newcommand{\ska}{SKA$_a$~}

\newcommand{\x}{\mathbf{x}}
\newcommand{\X}{\mathbf{X}}
\newcommand{\y}{\mathbf{y}}
\newcommand{\Y}{\mathbf{Y}}
\newcommand{\z}{\mathbf{z}}
\newcommand{\Z}{\mathbf{Z}}

\providecommand{\eqref}[1]{(\ref{#1})}
\newtheorem{construction}{Construction}

\newcommand{\msa}{\mathsf{D}}

\newcommand{\pxyz}{$P_{XYZ}$ }

\title{A One-way Secret Key Agreement with  Security Against Active Adversaries}

\author{
    \IEEEauthorblockN{Somnath Panja\IEEEauthorrefmark{1}, Shaoquan Jiang\IEEEauthorrefmark{2}, Reihaneh Safavi-Naini\IEEEauthorrefmark{1}}
    \IEEEauthorblockA{\IEEEauthorrefmark{1}University of Calgary, Canada}
    \IEEEauthorblockA{\IEEEauthorrefmark{2}University of Windsor, Canada}
}

\date{}

\begin{document}

\maketitle
\noindent
\begin{abstract}
In a one-way secret key agreement (OW-SKA) protocol in source model, Alice and Bob have private samples of two correlated variables $X$ and $Y$ that are partially leaked to Eve through $Z$, and use a single message from Alice to Bob to obtain a secret shared key. We propose an efficient secure OW-SKA when the sent  message  can be tampered with by  an {\em active adversary}.  The construction follows the approach of an existing OW-SKA with security against passive adversaries, and uses a specially designed secure 
Message Authentication Code (MAC) that is secure when the key is partially leaked, to achieve security against active adversaries. We prove the secrecy of the established key and robustness of the protocol,  and discuss our results.

\end{abstract}

\section{Introduction}

Secret key agreement is a fundamental problem in cryptography:
Alice wants to share a  random string called {\em key} with Bob, such that a third party Eve who has access to the communication channel between them, has no information about the key. 
Information theoretic secret key agreement (SKA) was first proposed by  Maurer~\cite{Maurer1993} and Ahlswede~\cite{Ahlswede1993}.
In their model that is referred to as {\em the source model},
Alice,  Bob and Eve  have private samples of   random
variables (RVs) $ X,  Y$ and  $Z$ with a joint probability distribution $P_{XYZ}$ that is known to all parties.
The goal of Alice and Bob  is to obtain a  common secret key
by exchanging messages over a public and error
free channel. 
An important quality parameter of an SKA protocol  is  the length  $\ell$ of the
established secret key. 
In the setting that Alice, Bob and Eve's variables $\X,\Y$  and $\Z$ are $n$
independent realizations of $X,Y$ and $Z$ distributed according to $P_{XYZ}$, 
the {\em secret-key rate } is defined  as the maximal rate at which Alice and Bob generate a highly secret key where the rate is   (informally) given by $\ell/n$.
\remove{
A natural setting is
when Alice, Bob and Eve have access to $n$
independent realizations of their random variables $[X_1,X_2\cdots X_n]$, $[Y_1,Y_2\cdots Y_n]$ and $[Z_1,Z_2\cdots Z_n]$, denoted by $(\X, \Y, \Z)$.
The {\em secret-key rate } is defined  as the maximal rate at which Alice and Bob generate a highly secret key where the rate is measured relative to the number of realizations of RVs, and is  (informally) given by $\ell/n$.
}
\remove{

The{\em secret key (SK) capacity $C_s(X,Y,Y)$} is the highest achievable key rate when $n \rightarrow \infty$. The SK capacity of a general distribution
is a long-standing open problem. In a real
life application of SKA protocols, $n$,  the number of available
samples to the parties, is finite and the efficiency of the
protocol is captured by the rate and in finite-block regime.
}
We consider {\em one-way secret key agreement} (OW-SKA) 
where Alice sends  a single message over the public channel to Bob, to arrive at a shared key.
OW-SKA problem is important in practice because it
avoids interaction between Alice and Bob,
\remove{
that  in addition to longer time and more complexity to arrive at an established key, requires
stateful protocols  which would introduce additional vulnerabilities in the implementation.  
The problem is
also 
}
as well as being  theoretically interesting because of its  relation  to circuit
polarization and immunization of public-key encryption in complexity theory and cryptography \cite{Holenstein2006}. 

{\em Adversaries.}
SKAs were first studied with the assumption that Eve is passive and only observes the communication channel. 
\remove{
SKA in this case requires that when one party accepts, the other also accepts and with a high probability both generate the same key.
}
Maurer \cite{Maurer1997authencation} considered 
a more powerful  adversary who can eavesdrop and tamper with the communication. 
Against such adversaries  the protocol is required to establish a   secret key when the adversary is passive, and 
with probability $1-\delta$ 
Alice {\em or } Bob must detect the tampering, or a  key that is unknown to Eve  be established. It was proved \cite{Maurer1997authencation} that SKA with security against active adversaries exists only when certaing {\em simulatablity conditions}  are satisfied (see section \ref{ap:related}.)

\remove{
In \cite{Maurer1997authencation}
a more powerful  adversary was considered where
in addition to eavesdropping the adversary can modify or  block the 
transmitted message. This is a more realistic model for communication in practice. 
In this case the protocol must establish a   secret key when the adversary is passive, and if active, with probability $1-\delta$ 
Alice or Bob must detect the corruption, or a  key that is unknown to Eve  be established. It is proved \cite{Maurer1997authencation} that SKA with security against active adversaries exists only when some {\em simulatablity conditions}  are satisfied (see section \ref{ap:related}.) 
}
{\em In the following we use \skp and \ska to denote SKA with security against passive and active adversaries, respectively.}

{\em Constructions.}  There are a number of constructions   for OW-\skp  
that (asymptotically) achieve the highest possible secret-key rate for a given distribution  $P_{\X\Y\Z}$
\cite{holenstein2005one,
renes2013efficient,Chou2015a,sharif2020}. 
 It was proved (Theorem 9, \cite{Maurer1997authencation}) that  secret-key rate of \ska
 is the same as the secret-key rate of \skp 
 if secure key agreement is possible (simulatability conditions hold).
 Construction of protocols with security against active adversaries  however is less studied.
 %
  It was also shown \cite{Maurer1997authencation}, through a construction, that it is possible to provide 
 message authentication when Alice, Bob and Eve have private samples of $P_{\X\Y\Z}$. In \cite{MaurerW03b} a concrete construction of a MAC was given, when the key $X$ is partially leaked to Eve through $Z$ for a known $P_{\X\Z}$. 
These MACs can be used to provide protection against active adversaries in \ska. 



%
\vspace{-0.5em}
\subsection*{ Our Work}
 We propose an efficient OW-\ska 
 for the setting of 
$P_{\X\Y\Z}$.  The construction is based on  a previous construction of  OW-\skp 
\cite{sharif2020} that achieves the highest secret-key rate for a given distribution,
and employs two hash functions, $h$ and $h'$ that are used for reconciliation and key extraction, respectively. Security proof of the protocol determines  parameters of the hash functions in terms of distribution $P_{\X\Y\Z}$, number of samples, key length and  security and reliability parameters  of the system.
We modify the protocol in two steps.\\
1) We  modify the inputs to $h$ (and so the message sent from 
Alice to Bob) is modified.  Theorem~\ref{Thm:ikemotsecurity}  recalculates the parameters of the two hash functions to achieve security against passive adversary, and 
the 
key length. This will also give 
%
the key length of our final construction when the protocol succeed to establish a key against an active adversary. \\
2) Noting that one of the inputs to $h$ is Alice's private sample $\x$, one can see the hash function as effectively a keyed hash function. This MAC however is different from traditional MAC systems that use a shared secret key. In Section \ref{owskamac} we define information theoretically secure MAC in correlated randomness setting 
of $P_{\X\Y\Z}$,  with security against impersonation and substitution attacks, and 
design a MAC with provable  security for the case when $\X=\Y$  and the shared key is partially leaked through $\Z$.
Using this MAC, which is a keyed hash function, for the hash function $h$ in our protocol, gives us a secure \ska. In Theorem~\ref{mac2:ctxt} we prove robustness of the protocol against active adversaries. 

To our knowledge the formal definition of MAC in  correlated randomness setting of $P_{\X\Y\Z}$ is new, and our construction of MAC is a new efficient construction with proved concrete security for special case of $\X=\Y$, and so  
would be of independent interest.
 Other known constructions of MACs   for  the general setting of $P_{\X\Y\Z}$ are due to Maurer \cite{Maurer1997authencation} that uses the theory of {\em typical sequences}  to prove asymptotic security (proof is not publicly available), and   for the case of $\X=\Y$ due to Maurer et al. \cite{MaurerW03b}. In Section \ref{comparison}, we compare our MAC with the MAC due to \cite{MaurerW03b}. 

\subsection{Related work}
\label{ap:related}

Key establishment with information theoretic setting has been studied in Quantum Key Distribution protocols (QKD) where the correlated random variables of Alice and Bob are obtained using quantum communication \cite{BB84}, as well as 
{\em fuzzy extractor setting} \cite{eurocryptDodisRS04} 
where  $\X $ and $\Y $ are samples of the same random source with a bound on the {\em distance}  between the two. Commonly used distance functions are Hamming distance and set difference. 

Key establishment 
with security against active adversaries was studied in \cite{maurer2003authen1,maurer2003authen2,Renner2004exact,kanukurthi2009key}. Feasibility of  information theoretic 
\ska was
formulated by Maurer through  {\em simulatability} property that is defined as follows: $P_{\X\Y\Z}$ is $\X$-simulatable  by Eve if they can send $\Z$ through a simulated channel $P_{\hat{\X}|\Z}$ whose output $\hat{\X}$ has the same joint distribution with $\Y$ as $\X$.
One can similarly define  $\Y$-simulatability. SKA cannot be constructed if $P_{\X\Y\Z}$ is $\X$-simulatable or $\Y$-simulatable.

\vspace{1mm}
\noindent
{\bf Organization.}
Section~\ref{sec:pre} is preliminaries. 
Section~\ref{sec:ikem-inst} is the construction of secure OW-\skp. 
 Section \ref{robustness} 
 is the construction of MAC 
 and 
 its security. 
Section~\ref{robustnessanalysis} uses the MAC to construct a secure  \ska. 
Section~\ref{conclusion} concludes the paper. 
%
 
\section{Preliminaries}
\label{sec:pre}
\noindent
We use capital letters (e.g., $X$) to 
denote 
 random variables (RVs),  and  lower-case letters (e.g., $x$) for their instantiations. Sets are denoted by calligraphic letters,  (e.g. $\mathcal{X}$), and the size of $\mathcal{X}$
is denoted by $|\mathcal{X}|$.
We denote vectors using boldface letters; for example $\X=(X_1,\cdots,X_n)$ is a vector of $n$ RVs, and its realization is given by $\x=(x_1,\cdots,x_n)$.  $U_\ell$ denotes an RV with uniform distribution over $\{0,1\}^\ell, \ell \in \mathbb{N}$. If $X$ is a discrete RV, we denote its probability mass function (p.m.f)
  by $\mathrm{P}_X(x)=\mathsf{Pr}(X=x)$. The conditional p.m.f. of an RV $X$ given RV $Y$ is denoted as $\mathrm{P}_{X|Y}(x|y)=\mathsf{Pr}(X=x|Y=y)$. For two RVs $X$ and $Y$ defined over the same domain $\mathcal{L}$, the statistical distance between $X$ and $Y$ is given by ${\rm \Delta}(X,Y)=\frac{1}{2} \sum_{v\in {\cal L}} |\Pr[X=v]-\Pr[Y=v]|$. \emph{Shannon entropy} of an RV $X$ is denoted by $H(X)=-\sum_x\mathsf{P}_X(x)\log(\mathsf{P}_X(x))$.

The  \emph{min-entropy} of a random variable $X$ with p.m.f. $\mathrm{P}_X$  
is defined as
$H_{\infty}(X)= -\log (\max_{x} (\mathrm{P}_X({x})))$.
The \emph{average conditional min-entropy}  \cite{dodis2004fuzzy} of an RV $X$ given RV $Y$ is 
$\tilde{H}_{\infty}(X|Y)= -\log (\mathbb{E}_{{y} \leftarrow Y}\max_{x}\mathrm{Pr}({X=x}|{Y=y})).
$

We write $[x]_{i\cdots j}$ to denote the block from the $i$th bit to $j$th bit in  $x.$ We use universal hash function on the output of weakly random entropy source, together with a random seed, to generate an output that is close to uniformly distributed,  
as shown by Leftover Hash Lemma~\cite{impagliazzo1989pseudo}. 

\begin{definition} [Universal hash family]\label{defn:uhf}
    A family of hash functions $h:\mathcal{X} \times \mathcal{S} \to \mathcal{Y}$ is called a universal hash family if $\forall x_1,x_2 \in \mathcal{X}$, $x_1 \ne x_2 :$  $\pr[h(x_1,s)=h(x_2,s)] \le \frac{1}{|\mathcal{Y}|}$, where the probability is over the uniform choices of $s$ from $\mathcal{S}$.
\end{definition}

\vspace{-.5em}
\begin{definition}[Strong universal hash family]\label{defn:suhash}
A family of hash functions $h: \mathcal{X} \times \mathcal{S} \rightarrow \mathcal{Y}$ is called a strong universal hash family if $\forall x_1,x_2 \in \mathcal{X}$, $x_1 \ne x_2$, and for any $c,d \in \mathcal{Y} :$ $\pr[h(x_1,s)=c \wedge h(x_2,s)=d]=\frac{1}{|\mathcal{Y}|^2}$, where the probability is  over the uniform choices of $s$ from $\mathcal{S}$.  
\end{definition}

We 
use a variant of Leftover Hash Lemma, called Generalized Leftover Hash Lemma [\cite{DodisORS08}, Lemma 2.4] 
 that includes side information about the hash input, to prove security properties of our construction.

 \begin{lemma}[Generalized Leftover Hash Lemma~\cite{DodisORS08}]
 \label{glhl}
Assume a universal hash family $h: \mathcal{X} \times \mathcal{S} \rightarrow \{0,1\}^{\ell}$. Then for any two random variables $A$ and $B$, defined over  $\mathcal{X}$ and $\mathcal{Y}$ respectively, applying $h$ on $A$ can extract a uniform random variable of length $\ell$ satisfying: \\  $\Delta(h(A, S), S, B; U_\ell, S, B)\le \frac{1}{2}\sqrt{2^{-\tilde{H}_{\infty}(A|B)}\cdot 2^
\ell}$, where $S$ is chosen randomly from $\mathcal{S}$. 
 \end{lemma}

We now recall the definition of almost strong universal hash family~\cite{Stinson94}.
\vspace{-.5em}
\begin{definition}[Almost strong universal hash family~\cite{Stinson94}]\label{defn:suhash1}
A family of hash functions $h: \mathcal{X} \times \mathcal{S} \rightarrow \mathcal{Y}$ is called  $\epsilon$-almost strong universal hash family if $\forall x_1,x_2 \in \mathcal{X}$, $x_1 \ne x_2$, and for any $c,d \in \mathcal{Y}$, it holds that : (a) $\pr[h(x_1,s)=c]=\frac{1}{|\mathcal{Y}|}$ and  (b) $\pr[h(x_1,s)=c \wedge h(x_2,s)=d] \leq \frac{\epsilon}{|\mathcal{Y}|}$, where the probability is  over the uniform choices of $s$ from $\mathcal{S}$.  
\end{definition}
 
 The notion of {\em fuzzy min-entropy}  has been introduced in \cite{Fuller2020fuzzy} to estimate the  guessing probability of a value within distance $t$ of a sample value $x$ of a distribution $P_{\X}$.  In~\cite{Fuller2020fuzzy}, Fuller et al. used fuzzy min-entropy to compute length of the extracted key in presence of passive adversaries. We consider active adversaries who try to guess a point around a secret key. We use fuzzy min-entropy to compute the probability that an active adversary can correctly guess that point. We define fuzzy min-entropy of a sample value $\x$ corresponding to a joint distribution $P_{\X\Y}$. The adversary tries to guess $\x_1$ such that the inequality:  $-\log(P_{\X |\Y }(\x_1 |\y )) \le \nu$ holds, where $\y$ in the secret key of Bob, and $\nu$ is some predetermined value. That is, the adversary tries to guess $\x_1$ such that $P_{\X |\Y }(\x_1 |\y ) \ge 2^{-\nu}$. To have the maximum chance that the inequality $P_{\X |\Y }(\x_1 |\y ) \ge 2^{-\nu}$ holds, the adversary would choose the point $\x_1$ that maximizes the total probability mass of $\Y$ within the set $\{\y : P_{\X |\Y }(\x_1 |\y ) \ge 2^{-\nu}\}$. 
The {\it $\nu$-fuzzy min-entropy}~\cite{Fuller2020fuzzy} of an RV $\X$ with joint distribution $P_{\X \Y}$ is defined as \\

{\small
$H_{\nu,\infty}^{\mathsf{fuzz}}(\X)=-\log\big(\mathsf{max}_{\x} \sum_{\y:-\log (P_{\X|\Y}(\x|\y)) \le \nu}\pr[\Y= \y]\big).$
}
For a joint distribution $P_{\X\Y\Z}$, the {\it $\nu$-conditional fuzzy min-entropy} of $\X$ given $\Z$ is defined as
{\scriptsize
\begin{align}\nonumber
\tilde{H}_{\nu,\infty}^{\mathsf{fuzz}}(\X|\Z)=-\log\Big(\underset{\z \leftarrow \Z}{\mathbb{E}}\max_{\x}\sum_{\y:-\log (P_{\X|\Y}(\x|\y)) \le \nu}\pr[\Y=\y|\Z=\z]\Big).
\end{align}
}

 The following lemma gives both lower and upper bounds of $H_{\nu,\infty}^{\mathsf{fuzz}}(\X)$. 
 \begin{lemma}\label{fuzzyent}
Let the joint distribution of two RVs $\X$ and $\Y$ be denoted as $P_{\X\Y}$. Then the following properties hold.

(i) $H_\infty(\X) - \nu \le H_{\nu,\infty}^{\mathsf{fuzz}}(\X).$

(ii) Let $\max_{\y}P(\y)=P(\y_{\max})$ for some point $\y_{\max}$ in the domain of $\Y$, and let there exist a point $\x_{\y_{\max}}$ in the domain of $\X$ such that $-\log(P_{\X|\Y}(\x_{\y_{\max}}|\y_{\max})) \le \nu$, then $H_{\nu,\infty}^{\mathsf{fuzz}}(\X) \le H_\infty(\Y)$.
 \end{lemma}

\begin{proof}
 
 (i) Note that,
{\small
\begin{align} \nonumber
&-\log (P_{\X|\Y}(\x|\y)) \le \nu \implies \frac{P(\x,\y)}{P(\y)} \ge 2^{-\nu} \\\label{eqn:fuzzyminent}
&\implies P(\y) \le 2^{\nu}P(\x,\y).
\end{align}
}

Let $\mathsf{max}_{\x} \sum_{\y:-\log (P_{\X|\Y}(\x|\y)) \le \nu}P(\x,\y)$ occur at a point $\x=\x_{\mathsf{max}}$, then
{\small
\begin{align}\nonumber
&\mathsf{max}_{\x} \sum_{\y:P_{\X|\Y}(\x|\y) \ge 2^{-\nu}}P(\x,\y) \\\label{eqn:fuzzyminentropy2}
&=  \sum_{\y:-P_{\X|\Y}(\x_{\mathsf{max}}|\y) \ge 2^{-\nu}}P(\x_{\mathsf{max}},\y)
\end{align}
}

Now, 

{
 \small
 \begin{align} \nonumber
 &H_{\nu,\infty}^{\mathsf{fuzz}}(\X) \\\nonumber
 &=-\log\big(\mathsf{max}_{\x} \sum_{\y:-\log (P_{\X|\Y}(\x|\y)) \le \nu}\pr[\Y= \y]\big) \\\nonumber
 &=\log\big(\frac{1}{\mathsf{max}_{\x} \sum_{\y:-\log (P_{\X|\Y}(\x|\y)) \le \nu}\pr[\Y= \y]}\big) \\\nonumber
 &\ge \log\big(\frac{1}{\mathsf{max}_{\x} \sum_{\y:-\log (P_{\X|\Y}(\x|\y)) \le \nu}2^{\nu}P(\x,\y)}\big) \\\nonumber
 &\qquad\qquad\text{  (by equation.~\ref{eqn:fuzzyminent})} \\\nonumber
 &=\log\big(\frac{1}{2^{\nu}\mathsf{max}_{\x} \sum_{\y:-\log (P_{\X|\Y}(\x|\y)) \le \nu}P(\x,\y)}\big) \\\nonumber
&=\log\big(\frac{1}{\mathsf{max}_{\x} \sum_{\y:-\log (P_{\X|\Y}(\x|\y)) \le \nu}P(\x,\y)}\big) - \nu \\\nonumber
&=\log\big(\frac{1}{\sum_{\y:P_{\X|\Y}(\x_{\mathsf{max}}|\y) \ge 2^{-\nu}}P(\x_{\mathsf{max}},\y)}\big) - \nu \text{  (by equation.~\ref{eqn:fuzzyminentropy2})} \\\nonumber
&\ge \log\big(\frac{1}{P(\x_{\mathsf{max}})}\big) - \nu \\\nonumber
& \ge \log\big(\frac{1}{\mathsf{max}_{\x}P(\x)}\big) -\nu  \\\label{eqn:fuzzyenteqn1}
&= H_{\infty}(\X) -\nu
 \end{align}
}

(ii) Let $\mathsf{max}_{\x} \sum_{\y:-\log (P_{\X|\Y}(\x|\y)) \le \nu}\pr[\Y= \y]$ occur at a point $\x=\hat{\x}_{\max}$, then 
{\small
\begin{align} \nonumber
&\mathsf{max}_{\x} \sum_{\y:P_{\X|\Y}(\x|\y) \ge 2^{-\nu}}\pr[\Y= \y] \\\nonumber
&=\sum_{\y:P_{\X|\Y}(\hat{\x}_{\max}|\y) \ge 2^{-\nu}}\pr[\Y= \y]    \\\label{eqn:fuzzent2}
&\ge \sum_{\y:P_{\X|\Y}(\x_{\y_{\max}}|\y) \ge 2^{-\nu}}\pr[\Y= \y]
\end{align}
}

{\small
\begin{align} \nonumber
&H_{\nu,\infty}^{\mathsf{fuzz}}(\X)     \\\nonumber
&=\log\big(\frac{1}{\mathsf{max}_{\x} \sum_{\y:-\log (P_{\X|\Y}(\x|\y)) \le \nu}\pr[\Y= \y]}\big) \\\nonumber
& \le \log\big(\frac{1}{\sum_{\y:P_{\X|\Y}(\x_{\y_{\max}}|\y) \ge 2^{-\nu}}\pr[\Y= \y]}\big) \text{  (by equation~\ref{eqn:fuzzent2})} \\\nonumber
& \le \log\big(\frac{1}{P(\y_{\max})}\big) \\\nonumber
&= \log\big(\frac{1}{\max_{\y}P(\y)}\big) \\\label{eqn:fuzzyenty3}
&=H_{\infty}(\Y)  
\end{align}
}
\end{proof}

As a corollary, we obtain the following lemma.
\begin{lemma}\label{fuzzyentcor}
Let the joint distribution of three RVs $\X$, $\Y$ and $\Z$ be denoted as $P_{\X\Y\Z}$. Then the following properties hold.

(i) $\tilde{H}_\infty(\X|\Z) - \nu \le \tilde{H}_{\nu,\infty}^{\mathsf{fuzz}}(\X|\Z).$

(ii) Let $\max_{\y}P(\y)=P(\y_{\max})$ for some point $\y_{\max}$ in the domain of $\Y$, and let there exist a point $\x_{\y_{\max}}$ in the domain of $\X$ such that $-\log(P_{\X|\Y}(\x_{\y_{\max}}|\y_{\max})) \le \nu$, then $\tilde{H}_{\nu,\infty}^{\mathsf{fuzz}}(\X|\Z) \le \tilde{H}_\infty(\Y|\Z)$.

\end{lemma}

\subsection{One-way secret key agreement (OW-SKA)}\label{owskamac}
\vspace{-.2em}

One natural setting is that the 
probabilistic experiment that underlies  
$(X,Y,Z)$ is repeated $n$ times independently, and Alice, Bob and Eve privately receive realizations of the RVs $\X = (X_1,\cdots,X_n)$, $\Y = (Y_1,\cdots,Y_n)$ and $\Z =(Z_1,\cdots,Z_n)$ respectively, where \\
$P_{\X \Y \Z }(\x , \y , \z )=P_{\X \Y \Z}(x_1,\cdots,x_n,y_1,\cdots,y_n,z_1,\cdots,z_n)\\=\prod_{i=1}^n P_{XYZ}(x_i,y_i,z_i)$. This setting is considered  in  
Maurer's satellite scenario where a randomly generated string is received by Alice, Bob and Eve over independent noisy channels ~\cite{Maurer1993,Maurer1997authencation}. 
We note that for all $i=1,\cdots,n,$, the RVs $(X_i, Y_i, Z_i)$
%
have the underlying distribution $P_{X Y Z}$ and $P_{X_i Y_i Z_i} (x,y,z)=P_{X Y Z}(x,y,z)$, $\forall (x,y,z) \in {\cal X}\times{\cal Y} \times {\cal Z}$.


\vspace{-.5em}
\begin{definition}[Secure OW-SKA protocol~\cite{holenstein2005one}]\label{owska}
Let $X$ and $Y$ be two RVs over $\mathcal{X}$ and $\mathcal{Y}$ respectively. For shared key length $\ell$, a OW-SKA protocol on $\mathcal{X} \times \mathcal{Y}$ consists of two function families: a (probabilistic) function family \{$\tau_{\mathrm{Alice}}:\mathcal{X}^n \to \{0,1\}^\ell  \times \mathcal{C}$\} that outputs $k_A \in \{0,1\}^\ell$ and $c$; and a function family $\{\tau_{\mathrm{Bob}}: \mathcal{Y}^n \times \mathcal{C} \to \{0,1\}^\ell\}$ that outputs $k_B \in \{0,1\}^\ell$. A OW-SKA protocol on $\mathcal{X} \times \mathcal{Y}$ is secure for a probability distribution $P_{XYZ}$ over $\mathcal{X} \times \mathcal{Y} \times \mathcal{Z}$ if for $\ell \in \mathbb{N}$, the OW-SKA protocol establishes an  $(\epsilon,\sigma)$-secret key 
$k$ satisfying the following properties:

(i) (reliability)\quad $\pr(K_A=K_B=K) \ge 1 - \epsilon$

(ii) (security) \quad $\Delta(K,C,\Z;U_\ell,C,\Z) \le \sigma$,

where $K_A$, $K_B$, $K$, $C$ are RVs corresponding to $k_A$, $k_B$, $k$ and $c$ respectively, and the RV $\Z$ is the vector of $n$ instances of the RV $Z$ and is Eve's side information.  
The RV $U_\ell$ is uniformly distributed over $\{0,1\}^\ell$.
\end{definition}

{\color{red}

}
\vspace{-1em}
\begin{definition}(Secure information-theoretic one-time MAC in correlated randomness setting) \label{defn:mac}

Let $P_{\X \Y \Z} $ be a public distribution.
 An $(|\mathcal{S}|,P_{\X\Y\Z},|\mathcal{T'}|,\delta_{mac})$-information-theoretic one-time message authentication code is a triple of algorithms $(\mathsf{gen},\mathsf{mac},\mathsf{ver})$ where
 $\mathsf{gen}: P_{\X \Y \Z} \to (\x\y\z)$  samples $ P_{\X \Y \Z}$ and privately gives $\x$ and $\y$  to Alice and Bob, and leaks $\z$ to Eve, 
$\mathsf{mac}: \mathcal{X}^n \times \mathcal{S} \to \mathcal{T'}$ is the tag generation algorithm that maps an input message from the message set $\mathcal{S}$  
to a tag in the tag set $\mathcal{T'}$,
using the private input of Alice, $\x $, and
$\mathsf{ver}: \mathcal{Y}^n \times \mathcal{S} \times \mathcal{T'} \to \{acc, rej\}$ takes a message and tag pair, and outputs either accept (acc) or reject (rej)  using  Bob's private input $\y $. 

The MAC satisfies correctness and unforgeability properties defined as follows.\\
{\em Correctness }: For any choice of $s \in \mathcal{S}$, we have,
%
\vspace{-.3em}
 \begin{equation} 
 \footnotesize
\pr[ (\x\y\z)  \leftarrow \mathsf{gen}(P_{\X\Y\Z}),\mathsf{ver}(\y ,s,\mathsf{mac}(\x ,s))=acc \text{ $|$ }\Z=\z] =1
\vspace{-.3em}
 \end{equation}
 {\em $\delta_{ot}$-Unforgeability (one-time unforgeability):} 
For any $(\x\y\z)  \leftarrow \mathsf{gen}(P_{\X\Y\Z})$, we consider protection against two types of attacks,  \\
(i) $\delta_{imp}$-impersonation: for any message and tag pair $s'  \in \mathcal{S}$ and $t'  \in \mathcal{T'}$ chosen by the adversary, the following holds. 
\vspace{-.4em}
{\small
 \begin{equation}
  \pr[ \mathsf{ver}(\y,s',t')=acc \text{ $|$ } \Z=\z] \le \delta_{imp}, \\
  \vspace{-.3em}
 \end{equation}
 }
(ii) $\delta_{sub}$-substitution: for any observed message and tag pair $(s,t)$, for any adversary choice of $s' \ne s  \in \mathcal{S}$ and $t', t  \in \mathcal{T'}$, the following holds.
\vspace{-.4em}
{\small
\begin{equation}
  \pr[ \mathsf{ver}(\y,s',t')=acc\text{ $|$ }
 s, \mathsf{mac}(\x ,s)=t, \Z =\z] \le \delta_{sub}, \\
 \vspace{-.3em}
 \end{equation}
 }
The MAC is called $\delta_{ot}$-one-time unforgeable with $\delta_{ot}=\mathsf{max}\{\delta_{imp},\delta_{sub}\}$, where the probability is taken over the randomness of  $P_{\X\Y\Z}$.

\end{definition}
\vspace{-.5em}
A special case of the above definition is when $\X=\Y$ and $P_{\X\Y\Z} = P_{\X\Z}$, where $\X$ is the shared key of Alice and Bob that is partially leaked through $\Z$ to Eve.

Maurer's  construction in \cite{Maurer1997authencation} is an example of the general case one-time MAC, while the construction in \cite{MaurerW03b} is an example of the  latter case $\X=\Y$ with partially leaked $\X$.

In the following we define robustness of a secure \ska.
We  follow the definition in \cite{Wolf98}.
\vspace{-.5em}
\begin{definition}[Robustness of secure OW-SKA protocol]
\label{robustowska} 
 A 
 OW-SKA protocol $(\tau_{\mathrm{Alice}}, \tau_{\mathrm{Bob}})$ is called $(\epsilon,\sigma)$ OW-\ska with robustness $\delta$ 
 if for any strategy of an active attacker with access to $Z$ and communicated messages  over the public channel, the probability that either Bob rejects the outcome of the protocol or the secret key agreement protocol is successful with reliability parameter $\epsilon$ and key security parameter $\sigma$, is no less than $(1-\delta)$. An  $(\epsilon,\sigma)$ OW-\ska protocol has robustness $\delta$ if for all adversary $\mathcal{D}$, the probability that the following experiment outputs `success' is at most $\delta$: sample $(\x,\y,\z)$ according to the distribution $P_{\X \Y \Z}$; let $(k_A,c) \leftarrow \tau_{\mathrm{Alice}}(\x)$, and let $\Tilde{c} \leftarrow \mathcal{D}(\z, c)$; output  `success' if $\Tilde{c} \ne c$ and $\tau_{\mathrm{Bob}}(\y,\tilde{c}) \ne \perp$.
 
\end{definition}

\remove{
++++++++++++++++
\begin{definition}[Robustness of secure OW-SKA protocol~\cite{Wolf98,dodis2006robust}]\label{robustowska} {\color{red} USE THE ONE THAT MATCHES}
 An $(\epsilon,\sigma)$-SK OW-SKA protocol has robustness $\delta$ if for every possible strategy of Eve, the probability that either Bob rejects the outcome of the protocol or the secret key agreement protocol is successful is no less than $(1-\delta)$. More precisely, an  $(\epsilon,\sigma)$-SK OW-SKA protocol has robustness $\delta$ if for all adversary $\mathcal{D}$, the probability that the following experiment outputs `success' is at most $\delta$: sample $(\x,\y,\z)$ according to the distribution $P_{\X \Y \Z}$, let $(k_A,c) \leftarrow \tau_{\mathrm{Alice}}(\x)$ and let $\Tilde{c} \leftarrow \mathcal{D}(\z, c)$; output  `success' if $\Tilde{c} \ne c$ and $\tau_{\mathrm{Bob}}(\y,\tilde{c}) \ne \perp$.
 
    
\end{definition}    
}
\section{Design and analysis of a OW-\ska protocol}
\label{sec:ikem-inst}
We give the construction of a robust and secure OW-\ska
protocol, and 
prove its security and robustness.
Our construction provides security against active adversary.

\subsection{A robust and secure OW-\ska 
}
\label{constructionikem}

A secure OW-SKA protocol that satisfies Definition~\ref{owska} provides security against passive adversary in which the adversary observes the message $c$ and tries to learn something about the extracted key $k$. However, it does not say anything about what happens if an adversary can modify a message $c$ as it is sent to Bob over public channel. In particular, the definition does not say anything about the output of $\tau_{\mathrm{Bob}}(\y,\tilde{c})$ if $\tilde{c} \ne c$.  A robust and secure \ska 
satisfies Definition~\ref{robustowska} and guarantees  that any tampering with $c$ will be either detected by Bob, or does not affect  a shared secret key establishment between Alice and Bob.
We build on 
the secure \skp 
protocol~\cite{sharif2020} and modify it to provide security and robustness.

The protocol uses 
two hash function families: a strong universal hash family $ h': \mathcal{X}^n \times \mathcal{S'} \rightarrow \{0,1\}^\ell $ and 
a universal hash family $ h: \mathcal{X}^n \times (\mathcal{S} \times \mathcal{S'})  \rightarrow \{0,1\}^t $ that are used to extract the key and construct the protocol message $c$, respectively.
$h$ is also an almost strong universal hash family when the probability is taken over $\mathcal{X}^n$. The message  and key domains are 
the sets $\mathcal{C} = \{0,1\}^t \times \mathcal{S} \times \mathcal{S'}$ and $\mathcal{K} = \{0,1\}^\ell $,
respectively. 
\vspace{-.3em}
\begin{construction}[A robust and secure OW-\ska] 
%
\label{owska:robust}

The \ska protocol $\mathsf{OWSKA}_{a}=(\ikemg,\ikeme,\ikemd)$ is given as follows: 

The \ska 
protocol $\mathsf{OWSKA}_{a}$ has three algorithms, $(\ikemg,\ikeme,\ikemd)$, that are given in Algorithm~\ref{alg:iKGen}, Algorithm~\ref{alg:iK.Enc} and Algorithm~\ref{alg:iK.dec}, respectively. The parameter $\nu$ in Algorithm~\ref{alg:iK.dec} relies on the correlation between the random variables $\X $ and $\Y $. Higher correlation between $\X $ and $\Y $ implies  smaller value of $\nu$ and smaller number of elements in the set $\mathcal{R}$. The relationship among parameters is given by both  Theorem~\ref{Thm:ikemotsecurity} and Theorem~\ref{mac2:ctxt}.
\end{construction}

\vspace{-1.5em}
 	\begin{algorithm}[!ht]
 	\footnotesize
 		\SetAlgoLined
 		\DontPrintSemicolon
 		\SetKwInOut{Input}{Input}\SetKwInOut{Output}{Output}
 		\Input{A public distribution $P_{\X \Y \Z }$}
 		\Output{$(\x ,\y ,\z )$}
 		\SetKw{KwBy}{by}
 		\SetKwBlock{Beginn}{beginn}{ende}
 		1. A trusted sampler samples the given public distribution $P_{\X \Y \Z }$ to generate $(\x ,\y ,\z ) \xleftarrow{\$} P_{\X \Y \Z }$ and provides them privately to Alice, Bob and Eve respectively. \\
 		\vspace{-1em}
 		\parbox{\linewidth}{\caption{\footnotesize $\ikemg(P_{\X \Y \Z })$}}  
 		\label{alg:iKGen}
 	\end{algorithm} 
\vspace{-2em}
 	\begin{algorithm}[!ht]
 	\footnotesize
 		\SetAlgoLined
 		\DontPrintSemicolon
 		\SetKwInOut{Input}{Input}\SetKwInOut{Output}{Output}
 		\Input{$\x $}
 		\Output{extracted key= $k$ and message=$c$}
 		\SetKw{KwBy}{by}
 		\SetKwBlock{Beginn}{beginn}{ende}
 		        1. Randomly sample seed $s' \xleftarrow{\$} \mathcal{S'}$ for $h'(\cdot)$ \\
 		        2. Randomly sample seed $s \xleftarrow{\$} \mathcal{S}$ for $h(\cdot)$\\
 		        3. $k$ = $h'(\x ,s')$ \\
 		        4. $c$ = $(h(\x , (s',s)),s',s)$ \\
 		        5. Output = $(k,c)$ \\
 		\vspace{-1em}
 		\parbox{\linewidth}{\caption{\footnotesize $\ikeme(\x )$}} 
 		\label{alg:iK.Enc}
 	\end{algorithm} 
\vspace{-2em}
 	\begin{algorithm}[!ht]
 	\footnotesize
 		\SetAlgoLined
 		\DontPrintSemicolon
 		\SetKwInOut{Input}{Input}\SetKwInOut{Output}{Output}
 		\Input{$\y $  and message $c$}
 		\Output{Either an extracted key $k$  or $\perp$}
 		\SetKw{KwBy}{by}
 		\SetKwBlock{Beginn}{beginn}{ende}
 		        1. Parse message $c$ as $(d,s',s)$, where $d$ is a $t$-bit string \\
 		        \vspace{-1.6em}
 		        \begin{flalign}\label{reconset}
 		        \text{2. Consider the set }\mathcal{R} =\{\x :-\log(P_{\X |\Y }(\x |\y )) \le \nu \} &&
 		        \end{flalign}
 		        3. For every $\hat{\x } \in \mathcal{R}$, Bob verifies whether $d=h(\hat{\x }, (s',s))$ \\
 		         4. \eIf{there is a unique $\hat{\x } \in \mathcal{R}$ satisfying  $d=h(\hat{\x }, (s',s))$}{
 		                   Output $k=h'(\hat{\x },s')$ \\
 		                   }
 		                   {
 		                   Output $\perp$ \\
 		                   }
 		\vspace{-.5em}
 		\parbox{\linewidth}{\caption{\footnotesize $\ikemd(\y ,c)$}} 
 		\label{alg:iK.dec}
 	\end{algorithm} 
\subsection{Relation with ~\cite{sharif2020}}
\label{securityfuzzyext}
 \vspace{-.3em}
 Our construction~\ref{owska:robust} is inspired by the construction of \skp in  ~\cite{sharif2020}, that is reproduced for completeness  
 in Appendix~\ref{app:theo}.
 The main difference of the constructions is that the  protocol message $c$ in  Construction 1  takes the seeds for both hash functions as inputs (in addition to $\x$). 
In Algorithm~\ref{alg:iK.Enc} of our construction, Alice executes a single hash computation $h(\x , (s',s))$ and sends the result together with the randomness $(s',s)$ to Bob. 
In Algorithm~\ref{alg:iK.dec} 
upon receiving the message, 
Bob searches the set $\mathcal{R}$ to find  a unique $\hat{\x } \in \mathcal{R}$ such that $h(\hat{\x }, (s',s))=h(\x , (s',s))$.
Robust information reconciliation 
succeeds if a unique $\hat{\x }$ is found in $\mathcal{R}$, allowing Bob 
to extract 
the key $k=h'(\hat{\x},s')$.
In  the \skp protocol in~\cite{sharif2020}, the  protocol message is $h(\x ,s)$ together with the randomness $(s',s)$. 
Bob's algorithm is similar but uses a different check function: Bob searches the set $\mathcal{R}$ to obtain a unique $\hat{\x }$ such that $h(\hat{\x } ,s)=h(\x ,s)$, and is successful if such an $\hat{\x }$ is found.
This change to the input of the hash function requires a complete re-evaluation of the  protocol parameters including parameters of the hash functions and the set $\cal R$, and security and robustness evaluation of the protocol.

 In Theorem~\ref{Thm:ikemotsecurity}, we prove that Construction \ref{owska:robust} is an \skp, 
 and  provide relationship among parameters and the length of the extracted key, 
In Theorem~\ref{mac2:ctxt}, we prove the robustness and show that the construction is an \ska.
 Combining these two theorems, we conclude that our construction is an $(\epsilon,\sigma)$-OW-\ska 
 protocol with robustness $\delta$ if the parameters $\ell$, $t$ and $\nu$ are chosen to satisfy both the Theorem~\ref{Thm:ikemotsecurity} and Theorem~\ref{mac2:ctxt}.

The following lemma~\ref{lemma:minentropy} can be proven using standard properties of independent variables and is given below.
\begin{lemma}\label{lemma:minentropy}  For any $(X_1Z_1), \cdots, (X_nZ_n)$ independently and identically distributed RV pairs, each with underlying distribution 
$P_{XZ}$,
 it holds that $\tilde{H}_\infty({\X }|{\Z })=n\tilde{H}_\infty(X|Z)$, where ${\X }=(X_1, \cdots, X_n)$ and ${\Z }=(Z_1, \cdots, Z_n).$ 
\end{lemma}


\begin{proof}
  $L_{\bf z}:=\max_{\bf x} P_{{\bf X}|{\bf Z}}({\bf x}|{\bf z})=\max_{\bf x}\prod_{i=1}^n P_{X_i|Z_i}(x_i|z_i)\\=\max_{\bf x}\prod_{i=1}^n P_{X|Z}(x_i|z_i).$ Note that ${\bf x}$ goes over ${\cal  X}^n$. Hence, $L_{\bf z}=\prod_{i=1}^n \max_{x_i} P_{X|Z}(x_i|z_i).$  We can define $x_z=\arg\max_x P_{X|Z}(x|z).$ Hence, 
$L_{\bf z}=\prod_{i=1}^n P_{X|Z}(x_{z_i}|z_i).$ Since 
\vspace{-1.5em}
{\small
\begin{align*} 
&\tilde{H}_\infty({\bf X}|{\bf Z})=-\log \sum_{{\bf z}\in {\cal Z}^n} (P_{\bf Z}({\bf z})L_{\bf z})\\
=&-\log \sum_{{\bf z}\in {\cal Z}^n}\prod_{i=1}^n P_{XZ}(x_{z_i}, z_i)\\
=&-\log \sum_{z_1, \cdots, z_n\in {\cal Z}}\prod_{i=1}^n P_{XZ}(x_{z_i}, z_i) \\
=&-\log \prod_{i=1}^n (\sum_{z_i\in {\cal Z}} P_{XZ}(x_{z_i}, z_i))\\
=&-\log \prod_{i=1}^n (\sum_{z\in {\cal Z}} P_{XZ}(x_{z}, z))
=-n\log (\sum_{z\in {\cal Z}} P_{XZ}(x_{z}, z))\\
=&-n\log (\sum_{z\in {\cal Z}} P_Z(z)P_{X|Z}(x_{z}|z))\\
=&-n\log (\sum_{z\in {\cal Z}} P_Z(z)\max_x P_{X|Z}(x|z))
=n\tilde{H}_\infty(X|Z).    
\end{align*} }
This completes the proof.  
\end{proof}

\subsection{Security against passive adversaries}

\begin{theorem}[Secure OW-SKA protocol]
\label{Thm:ikemotsecurity}
Let the parameters $\nu$ and $t$ be chosen such that \\
{\scriptsize $\nu = nH(X |Y ) + \sqrt(n) \log (|\mathcal{X}|+3) \sqrt{\log (\frac{\sqrt{n}}{(\sqrt(n)-1)\epsilon})}$,  and \\$t \ge nH(X |Y ) + \sqrt(n) \log (|\mathcal{X}|+3) \sqrt{\log (\frac{\sqrt{n}}{(\sqrt(n)-1)\epsilon})} + \log (\frac{\sqrt(n)}{\epsilon})$}, then the OW-\ska protocol  $\mathsf{OWSKA}_{a}$ given in construction~\ref{owska:robust} establishes a secret key of length $\ell \le  n\tilde{H}_{\infty}(X |Z ) + 2\log(\sigma) + 2 - t$ that is $\epsilon$-correct and $\sigma$-indistinguishable from random (i.e. $(\epsilon,\sigma)$-OW-SKA protocol according to Definition~\ref{owska}).
\end{theorem}


\begin{proof}
 We need to prove that the construction~\ref{owska:robust} satisfies Definition~\ref{owska} for secure OW-SKA protocol. We first prove the reliability of the protocol and then analyze its security.

{\it Reliability.} We first determine the value of $\nu$ and $t$ to bound the error probability (i.e. reliability) of the protocol by $\epsilon$, and then we compute the extracted secret key length $\ell$. In algorithm $\ikemd(\cdot)$~\ref{alg:iK.dec}, Bob searches the set $\mathcal{R}$ for $\hat{\x }$ and checks whether there is a unique $\hat{\x }$ whose hash value matches with the received hash value $d$. The algorithm succeeds if a unique $\hat{\x }$ is found in the set $\mathcal{R}$ with such property. Hence, the algorithm fails in the following two events: $(i)$ there is no element $\x \in \mathcal{R}$, whose hash value matches with the received hash value $d$ i.e. $\x \notin \mathcal{R}$, $(ii)$ the set $\mathcal{R}$ contains more than one element, whose hash values are the same as the received hash value $d$. Therefore, the probability that Bob fails to reproduce the correct key $\x $ is at most the sum of the probabilities of these two cases. These two cases represent the following two events respectively: 
\begin{align} \nonumber
 & \mathcal{E}_{1} = \{\x : \x  \notin \mathcal{R}\}=\{\x : -\log(P_{\X |\Y }(\x |\y )) > \nu\} \text{ and } \\\nonumber
 &\mathcal{E}_2 = \{\x  \in \mathcal{R}: \exists \text{ } \hat{\x } \in \mathcal{R} \text{ s.t. }  h(\x , (s',s)) = h(\hat{\x }, (s',s)\}.
 \end{align}
 
 For any $\epsilon > 0$, we choose $\epsilon_1 >0$ and $\epsilon_2 >0 $ satisfying $\epsilon_1 + \epsilon_2 \le \epsilon$. Let $\delta_1$ satisfy the equation \\$\epsilon_1 = 2^{\frac{-n{\delta_1}^2}{2\log^2(|\mathcal{X}|+3)}}$ and $\nu = H(\X |\Y ) + n\delta_1$.
 Then,  
 $\mathsf{Pr}(\mathcal{E}_{1})={\mathsf{Pr}}\left(-\log(P_{\X |\Y }(\x |\y )) > H(\X |\Y ) + n\delta_1\right) \le \epsilon_1$ (from~\cite{Holenstein11}, Theorem 2). 
 We now proceed to compute an upper bound for  $\mathsf{Pr}(\mathcal{E}_{2})$. Since $h(\cdot)$ is a  universal hash family, 
 for any $\x ,\hat{\x } \in \mathcal{R}$, $\x  \ne \hat{\x }$, random $s' \in \mathcal{S'}$ and randomly chosen $s \in \mathcal{S}$, we have $\mathsf{Pr}\left(h(\x ,( s',s))=h(\hat{\x },( s',s))\right) \le 2^{-t}$, where the probability is over the uniform choices of $(s',s)$ from $(\mathcal{S}' \times \mathcal{S})$. Consequently, $\mathsf{Pr}(\mathcal{E}_2) \le |\mathcal{R}|\cdot 2^{-t}$. From equation~\ref{reconset} and considering that the sum of probability of elements of the set $\mathcal{R}$ is less than or equal to 1, we obtain \\$\frac{|\mathcal{R}|}{2^\nu} \le \mathsf{Pr}(\mathcal{R}) \le 1 \Rightarrow |\mathcal{R}| \le 2^\nu$. Hence, $\mathsf{Pr}(\mathcal{E}_2) \le |\mathcal{R}|\cdot 2^{-t} \le 2^{\nu-t}.$ If we set $t=\nu -\log(\epsilon_2)$, we obtain  $\mathsf{Pr}(\mathcal{E}_2) \le \epsilon_2$. Thus, if $t=H(\X |\Y ) + n\delta_1 -\log(\epsilon_2)$, \\then the probability that Bob fails to reproduce the correct key $\x$ is at most $\mathsf{Pr}(\mathcal{E}_1) + \mathsf{Pr}(\mathcal{E}_2) \le \epsilon_1 + \epsilon_2 = \epsilon$. 
 Furthermore, considering that $\X , \Y $ are generated due to $n$ independent and identical experiments $P_{X_i Y_i Z_i}(x_i,y_i,z_i)=P_{X Y Z}(x_i,y_i,z_i)$ for $ 1 \le i \le n$, and $P_{\X \Y \Z }(\x ,\y ,\z )=\prod_{i=1}^n P_{XYZ}(x_i,y_i,z_i)$,  we have  $H(\X |\Y )=nH(X |Y )$.
 Now, setting \\$\epsilon_1=(\sqrt{n}-1)\epsilon/\sqrt{n}$ and $\epsilon_2=\epsilon/\sqrt{n}$, we have that if \\$\nu = nH(X |Y ) + \sqrt{n} \log (|\mathcal{X}|+3) \sqrt{\log (\frac{\sqrt{n}}{(\sqrt{n}-1)\epsilon})}$ and \\$t \ge nH(X |Y ) + \sqrt{n} \log (|\mathcal{X}|+3) \sqrt{\log (\frac{\sqrt{n}}{(\sqrt{n}-1)\epsilon})} + \log (\frac{\sqrt{n}}{\epsilon})$, then $\mathsf{Pr}(\mathcal{E}_1) + \mathsf{Pr}(\mathcal{E}_2) \le \epsilon$. Therefore, we conclude that the construction~\ref{owska:robust} is $\epsilon$-correct, and the reliability condition of Definition~\ref{owska} is satisfied.

{\it Security.} We now prove that the construction~\ref{owska:robust} also satisfies the security property of Definition~\ref{owska}. Let the RV $\Z $ correspond to $\z $, the attacker's initial information. Let $K$, $C$, $S'$ and $S$ be the RVs corresponding to the extracted key $k$, the ciphetext $c$, $s'$ and $s$ respectively, where $k=h'(\x ,s')$ and $c=\Big(h\big(\x , ( s',s)\big),s',s\Big)$. The RV $C$ is distributed over $\{0,1\}^t$.  Since $s',s$ are randomly chosen and independent of RV $\X $, from [\cite{DodisORS08}, Lemma 2.2(b)], we obtain $\tilde{H}_\infty(\X |\Z ,C) =\tilde{H}_\infty\big(\X |\Z ,h\big(\X , ( S',S)\big)\big)\ge \tilde{H}_\infty(\X |\Z ) - t$. 
Therefore, utilizing this expression, from Lemma~\ref{glhl}, we have 
{\scriptsize
\begin{align}\nonumber
&\Delta(K,C,\Z;U_\ell,C,\Z) \\\nonumber
&=\Delta\Big(h'(\X , S'), h\left(\X ,( S', S)\right), S', S, \Z ; U_\ell, h\left(\X ,(S', S)\right), S', S, \Z \Big) \\\nonumber
&\le \frac{1}{2}\sqrt{2^{-\tilde{H}_{\infty}(\X |\Z ,h\left(\X ,( S', S)\right))}\cdot 2^
{\ell}} \le \frac{1}{2}\sqrt{2^{-\tilde{H}_{\infty}(\X |\Z )}\cdot 2^
{\ell + t}} \\\nonumber
&=\frac{1}{2}\sqrt{2^{-n\tilde{H}_{\infty}(X |Z ) + \ell + t}} \le \sigma
\end{align}     
}

The last equality is obtained by applying  Lemma~\ref{lemma:minentropy} that proves $\tilde{H}_\infty({\X }|{\Z })=n\tilde{H}_\infty(X|Z)$. The last inequality follows due to $\ell \le n\tilde{H}_{\infty}(X |Z ) + 2\log(\sigma) + 2 - t$. 
 Consequently, the security property of Definition~\ref{owska} for OW-SKA protocol is satisfied. Therefore, construction~\ref{owska:robust} is an $(\epsilon,\sigma)$-OW-SKA protocol.  
 \end{proof}



\vspace{-.5em}
\section{Robustness of 
construction~\ref{owska:robust}} \label{robustness}
\vspace{-.3em}
We now prove that our construction~\ref{owska:robust} is a robust OW-\ska protocol as defined in Definition~\ref{robustowska}. In order to prove robustness, we consider specific construction of the almost strong universal hash family $h: \mathcal{X}^n \times (\mathcal{S} \times \mathcal{S'})  \rightarrow \{0,1\}^t$ described below.  

\vspace{-.3em}
\subsection{ A one-time secure MAC in $P_{\X\Z}$  setting} 
\vspace{-.3em}
For the private sample 
$\x $, we split $\x$ into two strings: \\$y'_1=[\x ]_{1\cdots t}$ and $y'_2=[\x ]_{t+1 \cdots n}$, where $t \le n/2$. Observe that $\x =y'_2 \parallel y'_1$.  
 For a message $m$, we represent it as a sequence $(s',s)$, where $s'=s''_2 \parallel s''_1$, $s''_1 \in GF(2^{n})$, $s''_2 \in GF(2^{n})$, $s=s_2\parallel s_1$, $s_1 \in GF(2^t)$ and $s_2 \in_R GF(2^{n - t})$  such that $s_2$  is non-zero on the last element. Note that we can always have $s_2$ to be non-zero on the last element by suitably appending a 1 to $m$. The verifier checks that the last element of $s_2$ is non-zero.
We represent $s'$, suitably padded with 1s, as a sequence $(s'_r,\cdots,s'_1)$ of elements of $GF(2^{n - t})$, where $r$ is odd. 
Define $h\big(\x ,(s',s)\big)=h\big(\x ,( s',(s_2,s_1))\big)=\big[s_2(y'_2)^{r+2} + {\sum}_{i=1}^r s'_i(y'_2)^i\big]_{1\cdots t} + (y'_1)^{3} + s_1y'_1$.  
We use this MAC to prove robustness of our construction in later section. Our MAC is inspired by the MAC in~\cite{CramerDFPW08}.

\vspace{-.3em}
\begin{lemma}
\label{lemma:mac2}
Let $h\big(\x , (s',s)\big)=\big[s_2(y'_2)^{r+2}    + {\sum}_{i=1}^r s'_i(y'_2)^i\big]_{1 \cdots t} + (y'_1)^{3} + s_1y'_1$ as defined in section~\ref{robustness}. Let $\mathsf{mac}:=h(\cdot)$. Define $\mathsf{ver}(\x ,(s',s), t)$ s.t. it outputs {\em acc} $\text{ if }h(\x , (s',s)) = t$ and {\em rej}, \text{ otherwise}. Then $(\mathsf{gen},h,\mathsf{ver})$ is an $(|\mathcal{S'} \times \mathcal{S}|,P_{\X\Z},|\mathcal{T'}|,\delta_{mac})$-information-theoretic one-time MAC with  $\delta_{imp}=3(r+2)2^{-(t+n\tilde{H}_{\infty}(X |Z )-n)}$, $\delta_{sub}=3(r+2)2^{-(t+n\tilde{H}_{\infty}(X |Z )-n)}$, and $\delta_{mac}=3(r+2)2^{-(t+n\tilde{H}_{\infty}(X |Z )-n)}$, where $\mathcal{T'}=\{0,1\}^t$. 
\end{lemma}

\begin{proof}
We need to prove that $(\mathsf{gen},h,ver)$ satisfies Definition~\ref{defn:mac} for the case $\x = \y$. Since $\x = \y $, it is easy to see that the correctness property of Definition~\ref{defn:mac} is satisfied. We now focus on the unforgeability property.
We first compute an adversary's success probability in an {\it{impersonation attack}} and then compute the adversary's success probability in an {\it{substitution attack}}.

{\it{Impersonation attack.}} In this attack, an adversary tries to generate a correct authenticated message $(t'_f,s'_f,s_f)$ such that
\vspace{-1em}
{\small
\begin{align}\label{eqn:impappn}
t'_f=\big[s_{2f}(y'_2)^{r+2} + {\sum}_{i=1}^r s'_{if}(y'_2)^i\big]_{1 \cdots t} + (y'_1)^{3} + s_{1f}y'_1, 
\vspace{-.6em}
\end{align}
}
where $s_f=s_{2f} \parallel s_{1f}$, $s'_f=(s'_{rf}\parallel \cdots \parallel s'_{1f})$, and the last element of $s_{2f}$ is non-zero. 
This is a non-zero polynomial in two variables $y'_2$ and  $y'_1$ of degree at most $(r+2)$. The term $s_{2f}(y'_2)^{r+2}+ {\sum}_{i=1}^r s'_{if}(y'_2)^i + 0^{n-t-t}|((y'_1)^{3} + s_{1f}y'_1)$ takes on each element in $GF(2^{n-t})$ for at most $3(r+2)2^t$ times when  $y'_2$ and $y'_1$ varies. Hence, there are at most $3(r+2)2^t (2^{n-t}/2^t)=3(r+2)2^{n-t}$ values of $(y'_2 \parallel y'_1)$ that satisfies the equation~\ref{eqn:impappn}. Let $\X$ and $\Z $ denote the RVs corresponding to $\x$ and $\z $ respectively.  Note that $\x = (y'_2 \parallel y'_1)$, and each value of $(y'_2 \parallel y'_1)$ (i.e. $\x$) occurs with probability at most $2^{-H_{\infty}(\X |\Z=\z )}$.   Thus, the probability that an adversary can successfully construct a correctly authenticated message = 

\vspace{-1.7em}
{\small
\begin{align}\nonumber
&\mathbb{E}_{\z  \leftarrow \Z}\Big[\mathsf{Pr}_{\X }\big[\mathsf{ver}(\x,(s'_f,s_f),t')=acc \text{ $|$ } \Z=\z\big]\Big] \\\nonumber
&=\mathbb{E}_{\z \leftarrow \Z }\Big[\mathsf{Pr}_{\X }\big[ t'_f=[s_{2f}(y'_2)^{r+2} + {\sum}_{i=1}^r s'_{if}(y'_2)^i]_{1 \cdots t} \\\nonumber
&\qquad + (y'_1)^{3} + s_{1f}y'_1 \text{ $|$ } \Z= \z\big]\Big] \\\nonumber
&\le \mathbb{E}_{\z  \leftarrow \Z}\big[3(r+2)2^{n-t}2^{-H_{\infty}(\X |\Z = \z)}\big] \\\nonumber
&= 3(r+2)2^{n-t}\mathbb{E}_{\z  \leftarrow \Z }\big[2^{-H_{\infty}(\X |\Z = \z)}\big]  \\\nonumber
&= 3(r+2)2^{n-t}\big[2^{-\tilde{H}_{\infty}(\X |\Z )}\big] 
=3(r+2)2^{-(t+\tilde{H}_{\infty}(\X |\Z ) - n)} \\\label{eqn:macimpappn}
&= 3(r+2)2^{-(t+n\tilde{H}_{\infty}(X |Z )-n)}.  
\end{align}
}
The last equality follows from Lemma~\ref{lemma:minentropy} that proves $\tilde{H}_{\infty}(\X |\Z )=n\tilde{H}_{\infty}(X |Z )$. Therefore, the success probability in an {\it impersonation attack} is at most $\delta_{imp}=3(r+2)2^{-(t+n\tilde{H}_{\infty}(X |Z )-n)}$.

{\it{Substitution attack.}}
Assume that an adversary is given a correctly authenticated message $c=(t',s',s)$, where 
{\small
\begin{align} \label{mac:senond1appn}
t'=\big[s_2(y'_2)^{r+2} + {\sum}_{i=1}^r s'_i(y'_2)^i\big]_{1 \cdots t } + (y'_1)^{3} + s_1 y'_1 .
\end{align}
}
Let $T'$, $C$, $Y'_1$ and $Y'_2$ denote the RVs corresponding to $t'$, $c$, $y'_1$ and $y'_2$ respectively. After observing the message $c=(t',s',s)$, the adversary tries to generate a forged message $(t'_f,s'_f,s_f)$ such that \\$t'_f=\big[s_{2f}(y'_2)^{r+2}+ {\sum}_{i=1}^r s'_{if}(y'_2)^i\big]_{1 \cdots t} +(y'_1)^{3} + s_{1f}y'_1$, $(t'_f,s'_f,s_f) \ne (t',s',s)$,  and the last element of $s_{2f}$ is non-zero. Then the expected probability that an adversary can successfully construct a forged message, given any message $c$, is \\ $\mathbb{E}_{(c, \z ) \leftarrow (C, \Z)}\big[\mathsf{Pr}_{\X }[ t'_f=\big[s_{2f}(y'_2)^{r+2} + {\sum}_{i=1}^r s'_{if}(y'_2)^i\big]_{1 \cdots t} +(y'_1)^{3} + s_{1f}y'_1 \text{ $|$ } C=c, \Z = \z ]\big]$. 
Now if $(s'_f,s_f)=(s',s)$, then Bob will reject the message unless $t'_f=t'$. Hence, we only need to focus on the case $(s'_f,s_f) \ne (s',s)$.

Since addition and subtraction correspond to the bit-wise exclusive-or in the corresponding field, we have 
\begin{align}\nonumber \label{mac:seconddiffappn}
t'-t'_f=&\big[(s_{2} - s_{2f})(y'_2)^{r+2}+ {\sum}_{i=1}^r (s'_{i} - s'_{if})(y'_2)^i\big]_{1 \cdots t} \\
&\quad + (s_1 - s_{1f})y'_1
\end{align}
    If $(s_1=s_{1f})$, then the degree of this polynomial in $y'_2$ is at most $(r+2)$. Now  the term\\ $[(s_{2} - s_{2f})(y'_2)^{r+2} + {\sum}_{i=1}^r (s'_{i} - s'_{if})(y'_2)^i\big]$ takes on each element of the field $GF(2^{n-t})$ at most $(r+2)$ times as $y'_2$ varies. Consequently, there are at most  $(r+2)(2^{n-t}/2^t)=(r+2)2^{n-2t}$ values of $y'_2$ that satisfies equation~\ref{mac:seconddiffappn}. 
Equation~\ref{mac:senond1appn} implies that, for each value of $y'_2$, there exists at most three values of $y'_1$ which satisfies the equation. Therefore, there are at most $3(r+2)2^{n-2t}$ values of $(y'_2 \parallel y'_1)$ that satisfies both the equation~\ref{mac:senond1appn} and equation~\ref{mac:seconddiffappn}. 
If $(s_1 \ne s_{1f})$, then representing $y'_1$ in equation~\ref{mac:seconddiffappn} in terms of $y'_2$ and substituting it in equation~\ref{mac:senond1appn}, we obtain \\$t'=[-(s_1 - s_{1f})^{-3}(s_{2} - s_{2f})^3 (y'_2)^{3(r+2)} ]_{_{1 \cdots t}} + g(y'_2)$ for some polynomial  $g(y'_2)$ of degree at most $3r$. Therefore, there are at most $3(r+2)2^{n-2t}$ values of $y'_2$ to satisfy this equation. From  equation~\ref{mac:seconddiffappn}, we see that, for each value of $y'_2$, there is a unique $y'_1$ that satisfies the equation. Therefore, in both the cases, there are at most $3(r+2)2^{n-2t}$ values of $(y'_2|| y'_1)$ that satisfies both the equation~\ref{mac:senond1appn} and equation~\ref{mac:seconddiffappn}. 
Note that $\x =y'_2 \parallel y'_1$. Then each value of  $(y'_2\parallel y'_1)$ (i.e. $\x $) occurs with probability at most $2^{-H_{\infty}(\X |\Z = \z ,T'=t')}$, where $\Z $ is the RV corresponding to $\z $, Eve's initial information.  Since $|t'|=t$, applying [\cite{DodisORS08}, Lemma 2.2(b)], we obtain  $\tilde{H}_{\infty}(\X |\Z ,T') \ge \tilde{H}_{\infty}(\X |\Z ) - t$.  

Therefore, the required expected probability =  
\vspace{-.5em}
{\small
\begin{align}\nonumber
&\mathbb{E}_{(c, \z ) \leftarrow (C, \Z)}\Big[\pr\big[ \mathsf{ver}(\x,(s'_f,s_f),t')=acc \\\nonumber
&\qquad\text{ $|$ } (s',s), 
 \mathsf{mac}(\x ,(s',s))=t, \Z =\z\big]\Big] \\\nonumber
&=\mathbb{E}_{(c, \z ) \leftarrow (C, \Z)}\Big[\mathsf{Pr}_{\X }\big[ t'_f=[s_{2f}(y'_2)^{r+2} + {\sum}_{i=1}^r s'_{if}(y'_2)^i]_{1 \cdots t}\\\nonumber
&\qquad +(y'_1)^{3} + s_{1f}y'_1 \text{ $|$ } C=c, \Z = \z \big]\Big] \\\nonumber
&=\mathbb{E}_{(c, \z ) \leftarrow (C, \Z )}\Big[\mathsf{Pr}_{(Y'_2\parallel Y'_1)}\big[ t'_f=[s_{2f}(y'_2)^{r+2} + \\\nonumber 
&\qquad{\sum}_{i=1}^r s'_{if}(y'_2)^i]_{1 \cdots t} +(y'_1)^{3} + s_{1f}y'_1 \\\nonumber
&\qquad\wedge t'=[s_2(y'_2)^{r+2} + {\sum}_{i=1}^r s'_i(y'_2)^i\big]_{1 \cdots t} +(y'_1)^{3} + s_{1}y'_1 \\\nonumber
&\qquad \text{ $|$ } C=c, \Z= \z \big]\Big]& \\\nonumber
&=\mathbb{E}_{(c, \z ) \leftarrow (C, \Z )}\Bigg[\mathsf{Pr}_{(Y'_2\parallel Y'_1)}\bigg[ t'-t'_f=\big[(s_{2} - s_{2f})(y'_2)^{r+2}+ \\\nonumber 
&\qquad  {\sum}_{i=1}^r (s'_{i} - s'_{if})(y'_2)^i\big]_{1 \cdots t}+ (s_1 - s_{1f})y'_1 \\\nonumber
&\qquad\wedge t'=\big[s_2(y'_2)^{r+2} + {\sum}_{i=1}^r s'_i(y'_2)^i\big]_{1 \cdots t} +(y'_1)^{3} + s_{1}y'_1 \\\nonumber
&\qquad \text{ $|$ } C=c, \Z = \z \bigg]\Bigg] \\\nonumber
&\le \mathbb{E}_{(c, \z ) \leftarrow (C, \Z )}\big[3(r+2)2^{n-2t}2^{-H_{\infty}(\X |\Z = \z , T'=t')}\big] \\\nonumber
&= 3(r+2)2^{n-2t}\mathbb{E}_{(c, \z ) \leftarrow (C, \Z )}\big[2^{-H_{\infty}(\X |\Z = \z ,T'=t')}\big]  \\\nonumber
&= 3(r+2)2^{n-2t}\big[2^{-\tilde{H}_{\infty}(\X |\Z ,T')}\big] 
\le 3(r+2)2^{n-2t}2^{-(\tilde{H}_{\infty}(\X |\Z )-t)} \\\label{eqn:macsubappn}
&= 3(r+2)2^{-(t+\tilde{H}_{\infty}(\X |\Z )-n)}= 3(r+2)2^{-(t+n\tilde{H}_{\infty}(X |Z )-n)}.  
\end{align}   
}
The last equality is obtained by using Lemma~\ref{lemma:minentropy} that proves $\tilde{H}_{\infty}(\X |\Z )= n\tilde{H}_{\infty}(X |Z )$.
Therefore, the success probability in a {\it substitution attack} is at most \\$\delta_{sub}=3(r+2)2^{-(t+n\tilde{H}_{\infty}(X |Z )-n)}$.

Consequently, $(\mathsf{gen},h,ver)$ is an $(|\mathcal{S'} \times \mathcal{S}|,P_{\X\Z},|\mathcal{T'}|,\delta_{mac})$\\-information-theoretic one-time MAC with  $\delta_{imp}=3(r+2)2^{-(t+n\tilde{H}_{\infty}(X |Z )-n)}$, $\delta_{sub}=3(r+2)2^{-(t+n\tilde{H}_{\infty}(X |Z )-n)}$, and $\delta_{mac}$$=\mathsf{max}\{\delta_{imp},\delta_{sub}\}$$=3(r+2)2^{-(t+n\tilde{H}_{\infty}(X |Z )-n)}$.  
\end{proof}

\subsubsection{Comparison with other MAC constructions}\label{comparison} 
Comparing our  results with 
\cite{MaurerW03b}, we note that
the MAC construction in [~\cite{MaurerW03b}, Theorem 3] for $t=n/2$, has 
success probabilities in {\it impersonation } and  {\it substitution attacks} as (roughly) $2^{-(H_2(\X | \Z = \z)-n/2)/2}) \approx 2^{-(n\tilde{H}_\infty(X | Z)-n/2)/2}$ and $(3\cdot 2^{-(H_2(\X | \Z = \z)-n/2)/4}) \approx 3 \cdot (2^{-(n\tilde{H}_\infty(X | Z)-n/2)/4})$, respectively, where $H_2(\X):=-\log(\sum_{\x \in \mathcal{X}^n}P_{\X }(\x)^2)$, the R{\'{e}}nyi entropy of $\X$.
In our construction however, the probabilities are significantly less as 
Lemma~\ref{lemma:mac2} shows that the success probabilities of the two attacks are the same, 
and are
at most $3(r+2)2^{-(t+n\tilde{H}_{\infty}(X |Z )-n)}$  which for $t=n/2$, 
are bounded by $3(r+2)2^{-(n\tilde{H}_{\infty}(X |Z )-n/2)}$. 

We note that the construction in ~\cite{MaurerW03b}  focuses on privacy amplification and assumes $\x=\y$, while  our construction  uses the MAC for reconcilliation, also.  Assuming $\x=\y$  
in our protocol,  
the extracted key  length will be roughly $(2n\tilde{H}_{\infty}(X |Z ) - n)$, and the extraction is possible as long as $n\tilde{H}_{\infty}(X |Z ) > n/2$ i.e. $\tilde{H}_{\infty}(X |Z ) > 1/2$. These  improve  the results of Maurer et al.[~\cite{MaurerW03b}, Theorem 5] that require $H_2(\X|\Z=\z) > 2n/3$, i.e. roughly $n\tilde{H}_{\infty}(X |Z ) > 2n/3$, and extracts a key of length roughly  $H_2(\X|\Z=\z) - 2n/3$, i.e. approximately $n\tilde{H}_{\infty}(X |Z ) - 2n/3$.

\subsubsection{Robustness analysis}\label{robustnessanalysis}
In order to prove robustness, we consider specific constructions of the strong universal hash family $h': \mathcal{X}^n \times \mathcal{S'} \rightarrow \{0,1\}^\ell$ and the universal hash family \\$h: \mathcal{X}^n \times (\mathcal{S} \times \mathcal{S'})   \rightarrow \{0,1\}^t$ as described below. $h$ is also an almost strong universal hash family when the probability is taken over the randomness of $\mathcal{X}^n$.  For the private sample 
$\x $, let $y'_1=[\x ]_{1\cdots t}$ and $y'_2=[\x ]_{t+1 \cdots n}$, where $t \le n/2$. Notice that $\x =y'_2 \parallel y'_1$.  
 For two random elements $s''_1,s''_2 \in_R GF(2^{n})$, let $s'=s''_2 \parallel s''_1$ and $h'(\x ,s')=h'\big(\x ,(s''_2,s''_1)\big)=\big[s''_2 (\x )^2+s''_1 \x \big]_{1 \cdots \ell}$.  We represent $s'$, suitably padded with 1s, as a sequence $(s'_r,\cdots,s'_1)$ of elements of $GF(2^{n - t})$ such that $r$ is odd. For two random elements $s_1 \in_R GF(2^{t})$, $s_2 \in_R GF(2^{n - t})$  such that the last element of $s_{2}$ is non-zero, let $s=s_2\parallel s_1$ and $h\big(\x ,(s',s)\big)=h\big(\x ,( s',(s_2,s_1))\big)=\big[s_2(y'_2)^{r+2} + {\sum}_{i=1}^r s'_i(y'_2)^i\big]_{1\cdots t} + (y'_1)^{3}+ s_1 y'_1$.  While verifying the tag, the receiver checks that the last element of $s_{2}$ is non-zero. 
Notice that $h$ and $h'$ are universal hash family and strong universal hash family respectively.  
 We now prove the following lemma~\ref{lemma:fuzzy} that will be used to prove robustness of our protocol. It is analogous to the Lemma 3.5 of Fuller et al.~\cite{Fuller2020fuzzy}. 
\begin{lemma}\label{lemma:fuzzy}
If $A$ is a random variable over a set of at most $2^b$ possible values, then  
$\tilde{H}_{\nu,\infty}^{\mathsf{fuzz}}(\X | A) \ge \tilde{H}_{\nu,\infty}^{\mathsf{fuzz}}(\X ) - b$.
\end{lemma}

\begin{proof}
{\footnotesize
\begin{align}\nonumber
&\tilde{H}_{\nu,\infty}^{\mathsf{fuzz}}(\X | A) \\\nonumber
&=-\log\Big(\underset{a \leftarrow A}{\mathbb{E}}\max_{\x}\sum_{\y:-\log (P_{\X|\Y}(\x|\y)) \le \nu}\pr[\Y=\y|A=a]\Big) \\\nonumber
&=-\log\Big(\sum_a\max_{\x}\sum_{\y:P_{\X|\Y}(\x|\y)) \ge 2^{-\nu}}\pr[\Y=\y|A=a]\pr[A=a]\Big) \\\nonumber
&=-\log\Big(\sum_a\max_{\x}\sum_{\y:P_{\X|\Y}(\x|\y)) \ge 2^{-\nu}}\pr[\Y=\y \wedge A=a]\Big) \\\nonumber
&\ge -\log\Big(\sum_a\max_{\x}\sum_{\y:P_{\X|\Y}(\x|\y)) \ge 2^{-\nu}}\pr[\Y=\y ]\Big) \\\nonumber
&\ge -\log\Big(2^b \max_{\x}\sum_{\y:P_{\X|\Y}(\x|\y)) \ge 2^{-\nu}}\pr[\Y=\y ]\Big) \\\nonumber
&\ge -\log\Big( \max_{\x}\sum_{\y:-\log (P_{\X|\Y}(\x|\y)) \le \nu}\pr[\Y=\y ]\Big) - b \\\nonumber
&\ge \tilde{H}_{\nu,\infty}^{\mathsf{fuzz}}(\X ) - b  \qquad\qquad\qquad\qquad\qquad\qquad\qquad\qquad 
\end{align}
}   
\end{proof}

\vspace{-1em}
\begin{theorem}[robustness of secure OW-SKA protocol]\label{mac2:ctxt}
 The robustness of secure OW-SKA protocol as defined in Definition~\ref{robustowska} is broken with probability at most  \\$3(r+2)\big( 2^{-(t-n+\min\{n\tilde{H}_{\infty}(X |Z ),\tilde{H}_{\nu,\infty}^{\mathsf{fuzz}}(\X |\Z )\})}\big)$. Hence, if $t \ge n + \log\big(\frac{3(r+2)}{\delta}\big) - \min\{n\tilde{H}_{\infty}(X |Z ),\tilde{H}_{\nu,\infty}^{\mathsf{fuzz}}(\X |\Z )\}$, then 
 the OW-\ska protocol $\mathsf{OWSKA}_{a}$ given in construction~\ref{owska:robust} has robustness $\delta$.
\end{theorem}

\begin{proof}
We need to prove that the construction~\ref{owska:robust} satisfies Definition~\ref{robustowska}. The algorithm $\ikemd(\cdot)$~\ref{alg:iK.dec} successfully outputs  an extracted key if there is a unique element $\hat{\x }$ in the set $\mathcal{R}$ such that $h(\hat{\x},s',s)$ is equal to the received hash value $d=h(\x,s',s)$. 
Note that, if the construction is $(\epsilon,\sigma)$-OW-SKA protocol, $\x$ is in the set  $\mathcal{R}$ with probability at least $(1 - \epsilon)$. 
In the robustness experiment as defined in Definition~\ref{robustowska}, an adversary receives an 
authenticated message $c$ with   $c=(t',s',s)$, where 
{\footnotesize 
\begin{align}
\nonumber 
t'=&h\big(\x , (s',s)\big)\\
=&\big[s_2(y'_2)^{r+2} + {\sum}_{i=1}^r s'_i(y'_2)^i\big]_{1 \cdots t} + (y'_1)^{3}+ s_1 y'_1, \label{eqn:robxappn}
\end{align}
}
 and the last element of $s_{2}$ is non-zero.
Define $f(y'_2,(s',s))=s_2(y'_2)^{r+2} + {\sum}_{i=1}^r s'_i(y'_2)^i$. For fixed $s, s'$, we denote the RVs corresponding to $t'$, $c$,  $y'_1$, $y'_2$, $\x $ as $T'$, $C$,  $Y'_1$, $Y'_2$, $\X $, respectively and hence the randomness is over  ${\X\Y\Z}$ only. After observing 
the message $c$ corresponding to $\x $, the adversary tries to generate  $c'$ 
corresponding to 
$\x $ or, some 
$\x_1 \in \mathcal{R}$ such that $\x_1 \ne \x $.
     Let $\delta_{\x}$ and $\delta_{\x_1}$ denote the success probabilities of the adversary in generating a $c'$ corresponding to the above two cases, respectively. 
 Since Bob's algorithm looks for a unique element in $\mathcal{R}$, only one of the above two cases succeed and hence, 
the probability that an adversary can generate a forged message is at most $\mathsf{max}\{\delta_\x,\delta_{\x_1}\}$, and we have
the success probability of the adversary bounded as: 
{\small
\begin{align}\nonumber
&\le \text{Expected probability that an adversary can generate a forged} \\\nonumber 
&\quad\quad\text{authenticated message given the message  $c=(t',s',s)$} \\\nonumber
&\quad\quad\text{corresponding to $\x $} \\\label{eq:forgetagappn}
&=\mathsf{max}\{\delta_{\x},\delta_{\x_1}\},
\end{align}
}
where the expectation is over $P_{\X\Y\Z}$.
We first bound $\delta_\x$. Let the adversary generate a forged message $(t'_f,s'_f,s_f)$ corresponding to $\x $ such that \\$t'_f=\big[s_{2f}(y'_2)^{r+2} + {\sum}_{i=1}^r s'_{if}(y'_2)^i\big]_{1 \cdots t} + (y'_1)^{3}+ s_{1f} y'_1$, $(t'_f,s'_f,s_f) \ne (t',s',s)$, and the last element of $s_{2f}$ is non-zero. Thus,
{\small
\begin{align}\nonumber
\delta_\x &\le \text{Expected probability that adversary can generate a forged} \\\nonumber
&\quad\quad\text{message $(t'_f,s'_f,s_f)$ valid when verified with $\x $, given } \\
\nonumber
&\quad\quad\text{the message  $c=(t',s',s)$} \\\nonumber
&=\mathbb{E}_{(c, \z) \leftarrow (C, \Z )}\Big[\mathsf{Pr}_{\X }\big[ t'_f=\big[s_{2f}(y'_2)^{r+2} + \\\nonumber
&\quad\quad {\sum}_{i=1}^r s'_{if}(y'_2)^i\big]_{1 \cdots t}+ (y'_1)^{3}+ s_{1f} y'_1 \text{ $|$ } C=c, \Z = \z \big]\Big] \\\label{eqn:forgedletax}
&\le 3(r+2)2^{-(t+n\tilde{H}_{\infty}(X |Z )-n)}
\end{align}
}
The expectation is over the distribution $P_{X|C=c, Z=z}$.
The inequality~\ref{eqn:forgedletax} follows from the {\it substitution attack} part of Lemma~\ref{lemma:mac2} (i.e.,  equation~\ref{eqn:macsubappn}). 

 We now compute the success probability (i.e., $\delta_{\x_1}$) that the forged value of $c$ 
 $(t'_{f_{\x_1}},s'_{f_{\x_1}},s_{f_{\x_1}})$ corresponds  to $\x_1 \neq \x $
such that \\ 
{\small
\begin{align} \nonumber
t'_{f_{\x_1}}&=h\big(\x_1 , (s'_{f_{\x_1}},s_{f_{\x_1}})\big)=\big[s_{2{f_{\x_1}}}(y'_{2{\x_1}})^{r+2} + \\\label{eqn:robx1}
&{\sum}_{i=1}^r s'_{i{f_{\x_1}}}(y'_{2{\x_1}})^i\big]_{1 \cdots t} + (y'_{1{\x_1}})^{3}+ s_{1{f_{\x_1}}} y'_{1{\x_1}}
\end{align}
}
and $(t'_{f_{\x_1}},s'_{f_{\x_1}},s_{f_{\x_1}}) \ne (t',s',s)$, where $y'_{1{\x_1}}=[\x_1 ]_{1\cdots t}$, $y'_{2{\x_1}}=[\x_1 ]_{t+1 \cdots n}$, $s_{f_{\x_1}}=s_{2{f_{\x_1}}} \parallel s_{1{f_{\x_1}}}$, $s'_{f_{\x_1}}=(s'_{r{f_{\x_1}}}\parallel \cdots \parallel s'_{1{f_{\x_1}}})$, and the last element of $s_{2{f_{\x_1}}}$ is non-zero. Since addition and subtraction correspond to bit-wise exclusive-or, we obtain,
{\small
\begin{align}\nonumber
t'-t'_{f_{\x_1}}&=\Big[\big[s_2(y'_2)^{r+2}+ {\sum}_{i=1}^r s'_i(y'_2)^i\big]_{1 \cdots t} + (y'_1)^{3}+ s_1 y'_1 \Big]  -\\\nonumber
&\qquad \Big[\big[s_{2{f_{\x_1}}}(y'_{2{\x_1}})^{r+2} + {\sum}_{i=1}^r s'_{i{f_{\x_1}}}(y'_{2{\x_1}})^i\big]_{1 \cdots t}\\\nonumber\label{eqn:robusttag1}
&\qquad + (y'_{1{\x_1}})^{3}+ s_{1{f_{\x_1}}} y'_{1{\x_1}}\Big] \\\nonumber
&=\big[f(y'_2,s',s) - f(y'_{2{\x_1}},s'_{f_{\x_1}},s_{f_{\x_1}})\big]_{1 \cdots t} + \\ 
&\qquad(y'_1)^{3} - (y'_{1{\x_1}})^{3} + s_1 y'_1 - s_{1{f_{\x_1}}} y'_{1{\x_1}}
\end{align}
}
Since $\x$ can be written as $\x_1 + e$ for some $e=(e_2 \parallel e_1) \in GF(2^{n})$, then $(y'_2\parallel y'_1)=((y'_{2{\x_1}} + e_2) \parallel (y'_{1{\x_1}} + e_1))$. Substituting $(y'_2\parallel y'_1)=((y'_{2{\x_1}} + e_2) \parallel (y'_{1{\x_1}} + e_1))$ in equation~\ref{eqn:robusttag1}, we have 
{\small
\begin{align}\nonumber
t'-t'_{f_{\x_1}}&=\Big[\big[s_2(y'_{2\x_1}+e_2)^{r+2}+ {\sum}_{i=1}^r s'_i(y'_{2\x_1}+e_2)^i\big]_{1 \cdots t} + \\\nonumber
&\qquad (y'_{1\x_1}+e_1)^{3}+ s_1 (y'_{1\x_1}+e_1) \Big]  -\\\nonumber
&\qquad \Big[\big[s_{2{f_{\x_1}}}(y'_{2{\x_1}})^{r+2} + {\sum}_{i=1}^r s'_{i{f_{\x_1}}}(y'_{2{\x_1}})^i\big]_{1 \cdots t}\\\label{eqn:robusttag111}
&\qquad + (y'_{1{\x_1}})^{3}+ s_{1{f_{\x_1}}} y'_{1{\x_1}}\Big]. 
\end{align}
}
 It is a polynomial in two variables $y'_{2{\x_1}}$ and $y'_{1{\x_1}}$ of degree at most $(r+2)$. 
 There may be two cases: either $e_1 \ne 0$ or $e_1=0$.

{\bf Case 1.} Let $e_1 \ne 0$, then equation~\ref{eqn:robusttag111} implies that there are at most two values of $y'_{1\x_1}$ for any fixed value of $y'_{2\x_1}$. 
From equation~\ref{eqn:robx1}, we obtain that, for each
value of $y'_{1\x_1}$, there exist at most $(r+2)2^{n-2t}$ values of $y'_{2\x_1}$ since $s_{2{f_{\x_1}}} \ne 0$. Therefore, there are at most $2(r+2)2^{n-2t}$ values of $(y'_{2\x_1} \parallel y'_{1\x_1})$ (i.e. $\x_1$) that satisfy both the equation~\ref{eqn:robxappn} and equation~\ref{eqn:robx1} (and hence also equation~\ref{eqn:robusttag111}).

{\bf Case 2.} Let $e_1=0$. There may be two sub-cases: either $e_2 \ne 0$ or $e_2=0$.

{\bf Subcase 2(i).} Let $e_2 \ne 0$, then equation~\ref{eqn:robusttag111} implies that there are at most $(r+2)2^{n-2t}$ values of $y'_{2\x_1}$ for any fixed value of $y'_{1\x_1}$ (note that $s_2 \ne 0$). From equation~\ref{eqn:robx1}, for each value of $y'_{2\x_1}$, we see that there exist at most three values of $y'_{1\x_1}$. Consequently, there are at most $3(r+2)2^{n-2t}$ values of $(y'_{2\x_1} \parallel y'_{1\x_1})$ (i.e. $\x_1$) that satisfy both the equation~\ref{eqn:robxappn} and equation~\ref{eqn:robx1} (and hence also equation~\ref{eqn:robusttag111}).

{\bf Subcase 2(ii).} Let $e_2 =0$. Then $y'_2=y'_{2{\x_1}}$ and $y'_1=y'_{1{\x_1}}$. Now proceeding the same way as {\it substitution attack} part of the proof of Lemma~\ref{lemma:mac2}, we prove that there are at most $3(r+2)2^{n-2t}$ values of $(y'_{2\x_1} \parallel y'_{1\x_1})$ (i.e. $\x_1$) that satisfy both the equation~\ref{eqn:robxappn} and equation~\ref{eqn:robx1} (and hence also equation~\ref{eqn:robusttag111}).

Therefore, in any case, there are at most $3(r+2)2^{n-2t}$ values of $(y'_{2{\x_1}} \parallel y'_{1{\x_1}})$ (i.e. $\x_1 $) that satisfy both the equation~\ref{eqn:robxappn} and equation~\ref{eqn:robx1} (and hence also equation~\ref{eqn:robusttag111}).

Let $\X_1$, $Y'_{2{\x_1}}$ and $Y'_{1{\x_1}}$ be the RVs corresponding to $\x_1$, $y'_{2{\x_1}}$ and $y'_{1{\x_1}}$ respectively.  We assume that the only way to attack the MAC and create a forged ciphertext is guessing the secret key for the MAC. The adversary tries to guess $\x_1$ such that the inequality:  $-\log(P_{\X |\Y }(\x_1 |\y )) \le \nu$ holds, where $\y$ in the secret key of Bob. That is, the adversary tries to guess $\x_1$ such that $P_{\X |\Y }(\x_1 |\y ) \ge 2^{-\nu}$. To have the maximum chance that the inequality $P_{\X |\Y }(\x_1 |\y ) \ge 2^{-\nu}$ holds, the adversary would choose the point $\x_1$ that maximizes the total probability mass of $\Y$ within the set $\{\y : P_{\X |\Y }(\x_1 |\y ) \ge 2^{-\nu}\}$. In addition, the adversary is given $\z$ and the authenticated message $c$. Therefore, an adversary can guess 
$\x_1$ with probability at most $2^{-H_{\nu,\infty}^{\mathsf{fuzz}}(\X | \Z = \z , T'=t')}$ (Lemma~\ref{fuzzyentcor} proves that guessing $\x_1$ given $\Z$ and $C$ in this way is better than guessing Bob's secret key $\y$ given $\Z$ and $C$).  Consequently,      
each value $\x_1$ in $\mathcal{R}$ occurs with probability at most $2^{-H_{\nu,\infty}^{\mathsf{fuzz}}(\X | \Z = \z , T'=t')}$.   
Since the size of the support of RV $T'$ is $2^t$, from Lemma~\ref{lemma:fuzzy}, we obtain  
$\tilde{H}_{\nu,\infty}^{\mathsf{fuzz}}(\X | \Z, T') \ge \tilde{H}_{\nu,\infty}^{\mathsf{fuzz}}(\X | \Z) - t$.

Thus,
{\small
\begin{align}\nonumber
\delta_{\x_1} &\le \text{Expected probability that an adversary can generate a forged} \\\nonumber
&\quad\quad\text{authenticated message $(t'_{f_{\x_1}},s'_{f_{\x_1}},s_{f_{\x_1}})$ corresponding to } \\\nonumber
&\quad\quad\text{$\x_1 $ given the message  $c=(t',s',s)$ corresponding to $\x $} \\\nonumber
&=\mathbb{E}_{(c, \z) \leftarrow (C, \Z )}\Big[\mathsf{Pr}_{(Y'_{2{\x_1}}\parallel Y'_{1{\x_1}}) }\big[ t'_{f_{\x_1}}=\big[s_{2{f_{\x_1}}}(y'_{2{\x_1}})^{r+2}+ \\\nonumber
&\quad\quad  \sum_{i=1}^r s'_{i{f_{\x_1}}}(y'_{2{\x_1}})^i\big]_{1 \cdots t} +  (y'_{1{\x_1}})^{3}+ s_{1{f_{\x_1}}} y'_{1{\x_1}}  \\\nonumber
&\quad\quad{\text{ $|$ } C=c, \Z = \z \big]\Big]} \\\nonumber
&=\mathbb{E}_{(c, \z) \leftarrow (C, \Z )}\Big[\mathsf{Pr}_{(Y'_{2{\x_1}}\parallel Y'_{1{\x_1}}) }\big[ t'_{f_{\x_1}}=\big[f(y'_{2{\x_1}},s'_{f_{\x_1}},s_{f_{\x_1}})\big]_{1 \cdots t}  \\\nonumber
&\qquad +  (y'_{1{\x_1}})^{3}+ s_{1{f_{\x_1}}} y'_{1{\x_1}} \\\nonumber 
&\qquad \wedge t' = \big[f(y'_2,s',s)\big]_{1 \cdots t} + (y'_1)^{3}+ s_1 y'_1 \text{ $|$ } C=c, \Z = \z \big]\Big] \\\nonumber
&\le \mathbb{E}_{(c, \z ) \leftarrow (C, \Z )}\big[3(r+2)2^{n-2t}2^{-H_{\nu,\infty}^{\mathsf{fuzz}}(\X |\Z = \z , T'=t')}\big] \\\nonumber
&= 3(r+2)2^{n-2t}\mathbb{E}_{(c, \z ) \leftarrow (C, \Z )}\big[2^{-H_{\nu,\infty}^{\mathsf{fuzz}}(\X |\Z = \z ,T'=t')}\big]  \\\nonumber
&= 3(r+2)2^{n-2t}\big[2^{-\tilde{H}_{\nu,\infty}^{\mathsf{fuzz}}(\X |\Z ,T')}\big] \\\nonumber
&\le 3(r+2)2^{n-2t}2^{-(\tilde{H}_{\nu,\infty}^{\mathsf{fuzz}}(\X |\Z )-t)} \\\label{eqn:robsub2}
&= 3(r+2)2^{-(t+\tilde{H}_{\nu,\infty}^{\mathsf{fuzz}}(\X |\Z )-n)}.
\end{align}
}



Consequently, from equations~\ref{eq:forgetagappn},~\ref{eqn:forgedletax} and~\ref{eqn:robsub2}, we conclude that, after observing a message, the expected probability that an adversary will be able to forge a message is at most \\$\max\{3(r+2)2^{-(t+n\tilde{H}_{\infty}(X |Z )-n)},3(r+2)2^{-(t+\tilde{H}_{\nu,\infty}^{\mathsf{fuzz}}(\X |\Z )-n)}\}$
$=3(r+2)\big( 2^{-(t-n+\min\{n\tilde{H}_{\infty}(X |Z ),\tilde{H}_{\nu,\infty}^{\mathsf{fuzz}}(\X |\Z )\})}\big)$ $ \le \delta$ (if $t \ge n + \log\big(\frac{3(r+2)}{\delta}\big) - \min\{n\tilde{H}_{\infty}(X |Z ),\tilde{H}_{\nu,\infty}^{\mathsf{fuzz}}(\X |\Z )\}$).
 Therefore, if $t \ge n + \log\big(\frac{3(r+2)}{\delta}\big) - \min\{n\tilde{H}_{\infty}(X |Z ),\tilde{H}_{\nu,\infty}^{\mathsf{fuzz}}(\X |\Z )\}$, the OW-\ska protocol $\mathsf{OWSKA}_{a}$ given in construction~\ref{owska:robust} has robustness $\delta$.  
\end{proof}

\vspace{-.5em}
\section{Concluding remarks.}\label{conclusion}
\vspace{-.3em}

 We proposed a OW-\ska in source model, proved its security and robustness, and derived the established key length.
To our knowledge there is no explicit construction of OW-\ska to compare our protocol with.
There are 
numerous OW-\skp constructions 
(Section~\ref{ap:related}) that can be the basis of  OW-\ska. 
Our construction in Section~\ref{constructionikem} 
is based on the construction in~\cite{sharif2020}, and uses a new
MAC construction (Section~\ref{robustness}).  
Interesting directions for future work will be improving efficiency of decoding (which is currently exponential), and 
proving capacity achieving property of the protocol.
\bibliographystyle{abbrv}
\bibliography{ska_arxiv}

\normalsize
\appendix


\section*{\skp protocol of Sharifian et al.~\cite{sharif2020}}\label{app:theo}

 	\begin{algorithm}[!ht]
 	\footnotesize
 		\SetAlgoLined
 		\DontPrintSemicolon
           
 		\SetKwInOut{Public}{Public}\SetKwInOut{Input}{Input}\SetKwInOut{Output}{Output}
            \Public{A publicly known distribution $P_{X Y Z}$}
 		\Input{$n$-component Alice's sample $\x \in \mathcal{X}^n$ and Bob's sample $\y \in \mathcal{Y}^n$, $\sigma$, $\epsilon$}
 		\Output{Alice's estimate of extracted key= $k_A$ and Bob's estimate of extracted key $k_B$}
 		\SetKw{KwBy}{by}
 		\SetKwBlock{Beginn}{beginn}{ende}
                    $\bullet$ Initialization phase: $(i)$ Alice and Bob compute and share $\ell$, $t$, and $\nu$ for the hash families $h(\cdot)$ and $h'(\cdot)$. $(ii)$ They generate  and share the seeds $s$ and $s'$ for the hash families $h$ and $h'$ respectively.\\
                    //Phase: Information reconciliation \\
 		        1. Alice computes the hash $g=h(\x , s)$ and sends it to Bob\\
                    2. Bob makes  a list of possible values of $\x$: \\
                       $M(\X|\y):=\{\hat{\x}:-\log(P_{\X|\Y}(\hat{\x }|\y )) \le \nu\}$ \\
                    3. Bob searches for $\hat{\x }$ in $M(\X|\y)$ such that $h(\x ,s)=g$ \\
                    4. \If {either $\hat{\x }$ is not unique in $M(\X|\y)$ or no $\hat{\x }$ is found in $M(\X|\y)$ with such property}{Abort}
                    //Phase: Key extraction \\
                    5. Alice and Bob compute the extracted key $k_A=h'(\x ,s')$ and $k_B=h'(\hat{\x },s')$ respectively. \\
 		\vspace{-1em}
 		\parbox{\linewidth}{\caption{\footnotesize $\Pi_{OWSKA}$: A capacity-achieving secret key agreement protocol}} 
 		\label{alg:SKA}
 	\end{algorithm}

We briefly recall the construction of secure \skp protocol due to Sharifian et al.~\cite{sharif2020}. The protocol is designed for source model in which Alice, Bob and Eve have $n$ components of the source $(\X ,\Y , \Z )$ respectively according to a distribution $P_{\X \Y \Z}(\x ,\y , \z )$. The protocol uses two universal hash families: $h: \mathcal{X}^n \times \mathcal{S} \to \{0,1\}^t$ and $h': \mathcal{X}^n \times \mathcal{S'} \to \{0,1\}^\ell$. The construction is given in Algorithm~\ref{alg:SKA}. This construction and its security theorems in~\cite{sharif2020} are given for the case that $P_{\X \Y \Z}(\x ,\y , \z )=\prod_{i=1}^n P_{X_i Y_i Z_i}(x_i,y_i,z_i)$ is due to $n$ independent experiments that are not necessarily identical.

\raggedbottom








\end{document}